\documentclass{article}%
\usepackage{amsfonts}
\usepackage{amsmath}
\usepackage{amssymb}
\usepackage{graphicx}%
\providecommand{\U}[1]{\protect\rule{.1in}{.1in}}
\newtheorem{theorem}{Theorem}

\newtheorem{problem}[theorem]{Problem}

\newenvironment{proof}[1][Proof]{\noindent\textbf{#1.} }{\ \rule{0.5em}{0.5em}}
\begin{document}

\title{On Quantum Algorithm for Binary Search and Its Computational Complexity}
\author{\dag S.Iriyama, \dag M.Ohya and \ddag I.V.Volovich\\\dag Tokyo University of Science, Japan\\\ddag Steklov Mathematical Institute, Russia}
\date{}
\maketitle

\begin{abstract}
A new quantum algorithm for a search problem and its computational complexity
are discussed. It is shown in the search problem containing $2^{n} $ objects
that our algorithm runs in polynomial time.

\end{abstract}

\section{Introduction}

Let $X$ and $Y$ be two finite sets and a function $f:X\rightarrow Y$. A search
problem is to find $x\in X$ such that \label{1}$f\left(  x\right)  =y$ for a
given $y\in Y$. There are two different cases for the search problem: (S1) one
is the case that we know there exists at least one solution $x$ of $f\left(
x\right)  =y$ in $X$. (S2) The other is the case that we do not know the
existence of such a solution. The second one is more difficult than the first
one. S1 belongs to a class NP, however S2 does to a class NP-hard\cite{GJ79}.

The search problem has been originally discussed by Levin\cite{L73,L84}, and
Solomonoff\cite{S84} described an algorithm of it. A quantum algorithm of the
search problem S1 was proposed by Grover in 1996\cite{G96}. The computational
complexity of Grover's searching algorithm is a square root of the cardinality
of $X$ denoted by $card\{X\}$.

In this paper, we studied the new quantum algorithm of the search problem S1
and S2 whose computational complexity is polynomial of $card\{X\}$. The idea
of this quantum algorithm is based on the amplification process of the OMV-SAT
algorithm\cite{OV1,OV2,OV3}.

\section{Search Problem}

Since S2 contains S1 as a special case, we will discuss S2 only here. A search
problem is defined by the following.

\begin{problem}
(S2) For a given $f$ and$\ y\in Y$, we ask whether there exists $x\in X$ such
that $f\left(  x\right)  =y$.
\end{problem}

Without loss of generality for discrete cases, we take $X=\left\{
0,1,\cdots,2^{n}-1\right\}  $ and $Y=\left\{  0,1\right\}  $. Let $M_{f,X,Y}$
be a Turing machine calculating $f\left(  x\right)  $ and checking whether
$f\left(  x\right)  =y$ with $x\in X$ and $y\in Y$. It outputs $1$ when
$f\left(  x\right)  =y$, $0$ otherwise. To solve this problem, one can
construct a Turing machine $M_{f}$ running as follows:

Step1: Set a counter $i=0$.

Step2: If $i>2^{n}-1$, then $M_{f}$ outputs "reject", else calls $M_{f,X,Y}$
with the inputs $x=i$ and $y$, so that $M_{f}$ obtains the result.

Step3: If the result of Step 2 is $1$, then it outputs $x$.

Step4: If the result is $0$, then it goes back to Step2 with the counter $i+1
$.

In the worst case, $M_{f}$ must call $M_{f,X,Y}$ for all $x$ to check whether
$f\left(  x\right)  =y$ or not, so that the computational complexity of the
searching algorithm is the cardinal number of $X$.

In the sequel sections, we construct a quantum algorithm to solve the problem
S2, and discuss on the computational complexity of it.

\section{Quantum Searching Algorithm}

From this section, we use a discrete function $f$. Let $n$ be a positive
number, and $f$ a function from $X=\left\{  0,1,\cdots,2^{n}-1\right\}  $ to
$Y=\left\{  0,1\right\}  $.

We show a quantum algorithm to solve the problem S2. To solve this problem, we
denote $x$ by the following binary expression
\begin{equation}
x=\sum_{k=1}^{n}2^{k-1}\varepsilon_{k},
\end{equation}
where $\varepsilon_{1},\cdots,\varepsilon_{n}\in\left\{  0,1\right\}  $.

We divide the problem S2 into several problems as below. Here we start the
following problem:

\begin{problem}
\label{prob1}Whether does there exist $x$ such that $f\left(  x\right)  =1$
with $\varepsilon_{1}=0$?
\end{problem}

If the answer is "yes", namely $\varepsilon_{1}=0$, then there exists at least
one $x=0\varepsilon_{2}\cdots\varepsilon_{n}$ such that $f\left(  x\right)
=1$. If the $\varepsilon_{1}\neq0$, then one considers two cases; the
$\varepsilon_{1}=1$, or there does not exist any $x$ such that $f\left(
x\right)  =1$.

We go to the next problem with the result of the above problem:

\begin{problem}
Whether does there exist $x$ such that $f\left(  x\right)  =1$ with
$\varepsilon_{2}=0$ for the obtained $\varepsilon_{1}$?
\end{problem}

After solving this problem, we know the value of $\varepsilon_{2}$, for
example, when $\varepsilon_{2}=0$, $x$ is written by $00\varepsilon_{3}%
\cdots\varepsilon_{n}$ or $10\varepsilon_{3}\cdots\varepsilon_{n}$.

Furthermore, we check the $\varepsilon_{i},$ $i=3,\cdots,n$ by the same way as
above using the information of the bits from $\varepsilon_{1}$ to
$\varepsilon_{i-1}$. We run the algorithm from $\varepsilon_{1}$ to
$\varepsilon_{n}$, and we look for one $x$ satisfying $f\left(  x\right)  =1$.
Finally in the case that the result of the algorithm is $x=1\cdots1$, we
calculate $f\left(  1\cdots1\right)  $ and check whether $f\left(
1\cdots1\right)  =1$ or not. We conclude that (1) if it becomes $1$,
$x=1\cdots1$ is a solution of search problem, and (2) otherwise, there does
not exist $x$ such that $f\left(  x\right)  =1$.

\section{Chaos Amplifier\label{CA}}

We will use the amplification process to construct the quantum searching
algorithm. For this purpose, in this section, let us review the Chaos
Amplifier along the papers \cite{OV1,OV2} and the book \cite{OV3}.

Consider the so called logistic map which is given by the equation
\begin{equation}
x_{n+1}=ax_{n}(1-x_{n})\equiv g_{a}(x),~~~x_{n}\in\left[  0,1\right]  .
\end{equation}

The properties of the map depend on the parameter $a.$ If we take, for
example, $a=3.71,$ then the Lyapunov exponent is positive, the trajectory is
very sensitive to the initial value and one has the chaotic behavior. It is
important to notice that if the initial value $x_{0}=0,$ then $x_{n}=0$ for
all $n.$

In the sequel sections, when we get the last qubit such that
\begin{equation}
\rho=\left(  1-p\right)  \left\vert 0\right\rangle \left\langle 0\right\vert
+p\left\vert 1\right\rangle \left\langle 1\right\vert ,
\end{equation}
one has to generate that an Abelian algebra by $\left\vert 0\right\rangle
\left\langle 0\right\vert $ and $\left\vert 1\right\rangle \left\langle
1\right\vert $ which can be considered as a classical system. If $p$ is very
small, e.g., $p=\frac{1}{2^{n}}$ with a large $n$, it is practically difficult
to distinguish $p=0$ and $p=\frac{1}{2^{n}}$, then we use the Chaos Amplifier
in the following manner.

Let $\Lambda_{CA}$ be a quantum channel on one qubit space such that
\begin{equation}
\Lambda_{CA}\left(  \rho\right)  =\frac{\left(  I+g_{a}\left(  \rho\right)
\sigma_{3}\right)  }{2},
\end{equation}
where $I$ is the identity matrix and $\sigma_{3}$ is the z-component of Pauli
matrices. Let $k$ be a positive integer, applying $\left(  \Lambda
_{CA}\right)  ^{k}$ to $\rho$, we have%
\begin{equation}
\left(  \Lambda_{CA}\right)  ^{k}\left(  \rho\right)  =\frac{(I+g_{a}%
^{k}(p)\sigma_{3})}{2}=\rho_{k}.
\end{equation}

To find a proper value $k$ we finally measure the value of $\sigma_{3}$ in the
state $\rho_{k}$ such that%

\begin{equation}
x_{k}\equiv\text{tr}\rho_{k}\sigma_{3}.
\end{equation}

The following theorems is proven in \cite{OV1,OV2,OV3}.

\begin{theorem}
For the logistic map $x_{n+1}=ax_{n}\left(  1-x_{n}\right)  $ with
$a\in\left[  0,4\right]  $ and $x_{0}\in\left[  0,1\right]  $, let $x_{0}$ be
$\frac{1}{2^{n}}$ and a set $J$ be $\left\{  0,1,2,\dots,n,\dots,2n\right\}
$. If $a$ is $3.71$, then there exists an integer $k$ in $J$ satisfying
$x_{k}>\frac{1}{2}.$
\end{theorem}

\begin{theorem}
Let $a$ and $n$ be the same in above theorem. If there exists $k$ in $J$ such
that $x_{k}>\frac{1}{2},$ then
\begin{equation}
k>\frac{n-1}{\log_{2}3.71-1}>\frac{5}{4}\left(  n-1\right)  .
\end{equation}

\end{theorem}

Using these theorems, we can easily check whether the state $\rho=\left\vert
0\right\rangle \left\langle 0\right\vert $ or not. Note that this
amplification process can be written in the generalized Turing machine form
\cite{IO1}, and it is related to the semigroup dynamics \cite{IO2}.

\section{Quantum Binary Searching Algorithm}

Let $m$ be a positive integer which can be written by a polynomial in $n$. Let
$\mathcal{H}=\left(  \mathbb{C}^{2}\right)  ^{\otimes n+m+1}$ be a Hilbert
space. The $m$ qubits are used for the computation of $f$, and the dust qubits
are produced by this computation. When $f$ is given, we can fix $m$. We will
show in the next section that this algorithm can be done in a polynomial time.

We construct the following quantum algorithm $M_{Q}^{\left(  1\right)  }$ to
solve the problem \ref{prob1}. Let $\left\vert \psi_{in}^{\left(  1\right)
}\right\rangle =\left\vert 0^{n}\right\rangle \otimes\left\vert 0^{m}%
\right\rangle \otimes\left\vert 0\right\rangle \in\mathcal{H}$ be an initial
vector for $M_{Q}^{\left(  1\right)  }$, where\ the upper index $\left(
1\right)  $ comes from the quantum algorithm checking the bit $\varepsilon
_{1}$. The last qubit of $\left\vert \psi_{in}^{\left(  1\right)
}\right\rangle $ is for the answer of it, namely "yes" or "no". If the answer
is "yes", then the last qubit becomes $\left\vert 1\right\rangle $, otherwise
$\left\vert 0\right\rangle $.

The quantum algorithm $M_{Q}^{\left(  1\right)  }$ is given by the following
steps. We start $M_{Q}^{\left(  1\right)  }$ with $\varepsilon_{1}=0$.

Step1: Apply Hadamard gates from the $2$nd qubit to the $n$-th qubit.%
\begin{align}
I\otimes U_{H}^{\otimes n-1}\otimes I^{m+1}\left\vert \psi_{in}^{\left(
1\right)  }\right\rangle  &  =\frac{1}{\sqrt{2^{n-1}}}\left\vert
\varepsilon_{1}\left(  =0\right)  \right\rangle \otimes\left(  \sum
_{i=0}^{2^{n-1}-1}\left\vert e_{i}\right\rangle \right)  \otimes\left\vert
0^{m}\right\rangle \otimes\left\vert 0\right\rangle \nonumber\\
&  =\left\vert \psi_{1}^{\left(  1\right)  }\right\rangle ,
\end{align}
where $\left\vert e_{i}\right\rangle $ are
\begin{align}
\left\vert e_{0}\right\rangle  &  =\left\vert 0\cdots0\right\rangle
\nonumber\\
\left\vert e_{1}\right\rangle  &  =\left\vert 1\cdots0\right\rangle
\nonumber\\
&  \vdots\nonumber\\
\left\vert e_{2^{n-1}-1}\right\rangle  &  =\left\vert 1\cdots1\right\rangle .
\end{align}

Let $U_{f}$ be the unitary operator on $\mathcal{H}=\left(  \mathbb{C}%
^{2}\right)  ^{\otimes n+m+1}$ to compute $f$, defined by%
\begin{equation}
U_{f}\left\vert x\right\rangle \otimes\left\vert 0^{m}\right\rangle
\otimes\left\vert 0\right\rangle =\left\vert x\right\rangle \otimes\left\vert
z_{x}\right\rangle \otimes\left\vert f\left(  x\right)  \right\rangle ,
\end{equation}
where $z_{x}$ is the dust qubit produced by the computation.

Step2: Apply the unitary operator $U_{f}$ to the state made in Step1, and
store the result in the last qubit.%
\begin{align}
U_{f}\left\vert \psi_{1}^{\left(  1\right)  }\right\rangle  &  =\frac{1}%
{\sqrt{2^{n-1}}}\left\vert 0\right\rangle \otimes\left(  \sum_{i=0}%
^{2^{n-1}-1}\left\vert e_{i}\right\rangle \otimes\left\vert z_{i}\right\rangle
\otimes\left\vert f\left(  0e_{i}\right)  \right\rangle \right) \nonumber\\
&  =\left\vert \psi_{2}^{\left(  1\right)  }\right\rangle ,
\end{align}
where $z_{i}$ is the dust qubits depending on $e_{i}$.

Step3: We take the last qubit by the projection from the final state
$\left\vert \psi_{2}^{\left(  1\right)  }\right\rangle $ such that
\begin{equation}
\left(  1-p\right)  \left\vert 0\right\rangle \left\langle 0\right\vert
+p\left\vert 1\right\rangle \left\langle 1\right\vert ,
\end{equation}
where $p=card\left\{  x|f\left(  x\right)  =1,x=0\varepsilon_{2}%
\cdots\varepsilon_{n}\right\}  /2^{n-1}$, and this state is the state $\rho$
given in the previous section.

Step4: After the above formula, the state is a pure state or a mixed state. If
the state is mixed and $p\neq0$ however very small, then apply the Chaos
Amplifier given in the Section\ref{CA} to check whether the last qubit is in
the state $\left\vert 1\right\rangle \left\langle 1\right\vert $. If we find
that the last qubit is in the state $\left\vert 1\right\rangle \left\langle
1\right\vert $, then $p\neq0$, which implies that there exists at least one
solution of $f\left(  x\right)  =1$ for $\varepsilon_{1}=0$. If we do not find
that the last qubit is in the state $\left\vert 1\right\rangle \left\langle
1\right\vert $, namely $p=0$, then there are two possibilities that are
$\varepsilon_{1}=1$ or no solutions $x\in X$ of $f\left(  x\right)  =1$.

After this algorithm, we know that if $\varepsilon_{1}=0$ or $1$, then the
last qubit is $1$ or $0$, respectively. We write this process as
$M_{Q}^{\left(  1\right)  }\left(  0^{n}\right)  =\varepsilon_{1}$ where
$0^{n}$ means the initial vector.

Next we modify Step1 of the algorithm $M_{Q}^{\left(  1\right)  }$ as:

Step1: Apply Hadamard gates from $3$rd qubit to $n$-th qubit.

And we call this algorithm $M_{Q}^{\left(  2\right)  }$. The index $\left(
2\right)  $ means that the algorithm check $\varepsilon_{2}$. We start
$M_{Q}^{\left(  2\right)  }$ with the initial vector $\left\vert \psi
_{in}^{\left(  2\right)  }\right\rangle =\left\vert \varepsilon_{1}%
,0^{n-1}\right\rangle \otimes\left\vert 0^{m}\right\rangle \otimes\left\vert
0\right\rangle $ instead of $\left\vert \psi_{in}^{\left(  1\right)
}\right\rangle $.

So forth we obtain the bit $\varepsilon_{2}$, and write as $M_{Q}^{\left(
2\right)  }\left(  \varepsilon_{1},0^{n-1}\right)  =M_{Q}^{\left(  2\right)
}\left(  M_{Q}^{\left(  1\right)  }\left(  0^{n}\right)  ,0^{n-1}\right)
=\varepsilon_{2}$.

In generally, we write the algorithm $M_{Q}^{\left(  i\right)  }\left(
\varepsilon_{1},\varepsilon_{2},\cdots,\varepsilon_{i-1},0^{n-i+1}\right)  $
for an initial vector $\left\vert \psi_{in}^{\left(  i\right)  }\right\rangle
=\left\vert \varepsilon_{1},\varepsilon_{2},\cdots,\varepsilon_{i-1}%
,0^{n-i+1}\right\rangle \otimes\left\vert 0^{m}\right\rangle \otimes\left\vert
0\right\rangle $ as the following:

Step1: Apply Hadamard gates from $i+1$-th to $n$-th qubits$.$%
\begin{align}
I^{\otimes i}\otimes U_{H}^{\otimes n-i}\otimes I^{m+1}\left\vert \psi
_{in}^{\left(  i\right)  }\right\rangle  &  =\frac{1}{\sqrt{2^{n-i}}%
}\left\vert \varepsilon_{1},\varepsilon_{2},\cdots,\varepsilon_{i-1}%
\right\rangle \otimes\left(  \sum_{k=0}^{2^{n-i}-1}\left\vert e_{k}%
\right\rangle \right)  \otimes\left\vert 0^{m}\right\rangle \otimes\left\vert
0\right\rangle \nonumber\\
&  =\left\vert \psi_{1}^{\left(  i\right)  }\right\rangle .
\end{align}

Step2: Apply the unitary gate to compute $f$ for the superposition made in
Step1, and store the result in $n+m+1$-th qubit.%
\begin{align}
U_{f}\left\vert \psi_{1}^{\left(  i\right)  }\right\rangle  &  =\frac{1}%
{\sqrt{2^{n-i}}}\left\vert \varepsilon_{1},\varepsilon_{2},\cdots
,\varepsilon_{i-1}\right\rangle \otimes\left(  \sum_{k=0}^{2^{n-1}%
-1}\left\vert e_{k}\right\rangle \otimes\left\vert z_{k}\right\rangle
\otimes\left\vert f\left(  \varepsilon_{1},\varepsilon_{2},\cdots
,\varepsilon_{i-1},e_{k}\right)  \right\rangle \right) \nonumber\\
&  =\left\vert \psi_{2}^{\left(  i\right)  }\right\rangle .
\end{align}

Step3: Take the last qubit by the projection from the final state $\left\vert
\psi_{2}^{\left(  i\right)  }\right\rangle $ such that%

\begin{equation}
\left(  1-p\right)  \left\vert 0\right\rangle \left\langle 0\right\vert
+p\left\vert 1\right\rangle \left\langle 1\right\vert .
\end{equation}

Step4: Apply the Chaos Amplifier to the amplitude $p$, so that we can easily
find that the last qubit is $\left\vert 1\right\rangle \left\langle
1\right\vert $.

After this algorithm $M_{Q}^{\left(  i\right)  }\left(  \varepsilon
_{1},\varepsilon_{2},\cdots,\varepsilon_{i-1},0^{n-i+1}\right)  $, we know the
bit $\varepsilon_{i}$ such that $f\left(  x\right)  =1$. Each $M_{Q}^{\left(
i\right)  }$, $i\geq2$ use the result of all $M_{Q}^{\left(  j\right)  }$
$\left(  j<i\right)  $ as an initial vector. We run this algorithm
$M_{Q}^{\left(  i\right)  }$ for each $i$ $\left(  i=1,\cdots,n\right)  $.

\section{Computational Complexity of the Quantum Binary Search Algorithm}

In this section, we calculate the computational complexity of the quantum
algorithm for binary search. The computational complexity is the number of
total unitary gates discussed above and amplification channels in our search algorithm.

In the above section, the quantum algorithm for binary search is given by the
products of unitary gates denoted by $U_{i}$ below. Let $\left\vert \psi
_{in}^{\left(  i\right)  }\right\rangle $ be an initial vector for the
algorithm $M_{Q}^{\left(  i\right)  }$ as%
\begin{equation}
\left\vert \psi_{in}^{\left(  i\right)  }\right\rangle =\left\vert
\varepsilon_{1}\cdots\varepsilon_{i-1},0^{n-i}\right\rangle \otimes\left\vert
0^{m}\right\rangle \otimes\left\vert 0\right\rangle ,
\end{equation}
and it goes to the final vector%
\begin{align}
U_{i}\left\vert \psi_{in}^{\left(  i\right)  }\right\rangle  &  =\frac
{1}{\sqrt{2^{n-i}}}\left\vert \varepsilon_{1},\cdots,\varepsilon
_{i-1}\right\rangle \otimes\left(  \sum_{k=0}^{2^{n-1}-1}\left\vert
e_{k}\right\rangle \otimes\left\vert z_{k}\right\rangle \otimes\left\vert
f\left(  \varepsilon_{1},\cdots,\varepsilon_{i-1},e_{k}\right)  \right\rangle
\right) \nonumber\\
&  =\left\vert \psi_{2}^{\left(  i\right)  }\right\rangle ,
\end{align}
where $f\left(  \varepsilon_{1},\cdots,\varepsilon_{i-1},e_{k}\right)  $ is
the result of the objective function for a search problem. The above unitary
gates $U_{i}$ for the algorithm $M_{Q}^{\left(  i\right)  }$ are defined by%
\begin{equation}
U_{i}=U_{f}\left(  U_{H}\right)  ^{\otimes n-i}%
{\displaystyle\prod\limits_{\left\{  x_{k}|x_{k}=1\right\}  }}
U_{NOT}\left(  k\right)  ,
\end{equation}
where $U_{NOT}\left(  k\right)  $ is to apply the NOT gate for the $k$-th
qubit only when the result of stage $k$ is $1$, $\left(  k=1,2,\cdots
i-1\right)  $ .

The computational complexity $T$ of the quantum binary search algorithm
$T\left(  U_{n}\right)  $ is given by the total number of unitary gates and
quantum channels for the amplification. We obtain the following theorem.

\begin{theorem}
We have
\begin{equation}
T=\left[  \frac{13}{8}n^{2}-\frac{9}{4}n+nT\left(  U_{f}\right)  \right]  +1,
\end{equation}
where $\left[  \cdot\right]  $ means the Gauss symbol, and $T\left(
U_{f}\right)  $ is a given complexity associated to the function $f$.
\end{theorem}

\begin{proof}
For the algorithm $M_{Q}^{\left(  i\right)  }$, one should have the following
gates: $i-1$ NOT gates, $n-i$ Hadamard gates and $U_{f}$, so that the
computational complexity $T_{i}$ for the algorithm $M_{Q}^{\left(  i\right)
}$ is given by%
\begin{align}
T_{i}  &  =n-i+i-1+T\left(  U_{f}\right) \nonumber\\
&  =n-1+T\left(  U_{f}\right)  .
\end{align}
The total number of stages is $n$, then the computational complexity is
\begin{equation}
\sum_{i=1}^{n}T_{i}=n\left(  n-1\right)  +nT\left(  U_{f}\right)  .
\end{equation}
For the amplification process explained above, the number of amplification
channels for $n$ qubits was shown as%
\begin{equation}
\left[  \frac{5}{4}\left(  n-1\right)  \right]  +1.
\end{equation}
In the algorithm $M_{Q}^{\left(  i\right)  }$, we have to apply $\left[
\frac{5}{4}\left(  n-1\right)  \right]  +1$ times amplification channels for
$\left(  n-i\right)  $ qubits at worst. Then the number of total quantum
channels is calculated as%
\begin{equation}
\frac{1}{2}n\left(  \frac{5}{4}\left(  n-1-1\right)  +0\right)  \leq\left[
\frac{5}{8}n\left(  n-2\right)  \right]  +1.
\end{equation}
Therefore the computational complexity $T$ becomes%
\begin{align}
T  &  =\sum_{i=1}^{n}T\left(  U_{i}\right)  +\left[  \frac{5}{8}n\left(
n-2\right)  \right]  +1\nonumber\\
&  =n\left(  n-1\right)  +nT\left(  U_{f}\right)  +\left[  \frac{5}{8}n\left(
n-2\right)  \right]  +1\nonumber\\
&  \leq\left[  \frac{13}{8}n^{2}-\frac{9}{4}n+nT\left(  U_{f}\right)  \right]
+1.
\end{align}

\end{proof}

Note that the above $T\left(  U_{f}\right)  $ is essentially polynomial in $n$.

\section{Conclusion}

In this paper, we constructed the quantum algorithm for searching probrem S2
for a given $f$ and$\ y\in Y$ with the cardinal number of $X;card\left\{
X\right\}  =2^{n}$. Our quantum algorithm can be written in a combination of
quantum algorithms $M_{Q}^{\left(  i\right)  }$ for $i=1,\cdots,n$. The each
quantum algorithms are run sequentially, in the other words, $M_{Q}^{\left(
i\right)  }$ uses the result of $M_{Q}^{\left(  i-1\right)  }$. This quantum
algorithm is able to check whether there exists a cirtain $x$ such that
$f\left(  x\right)  =y$ or not; it solved an NP-hard problem. We proved that
the computational complexity of our quantum searching algorithm is polynomial
in $n$.

\end{document}